\documentclass[12pt,oneside,reqno]{amsart}

\usepackage{mathrsfs, amsmath,amsfonts,amssymb}
\usepackage{amsmath}
\usepackage{mathtools}
\usepackage{braket}

\theoremstyle{plain}

\newtheorem{corollary}{Corollary}[section]

\newtheorem{definition}{Definition}[section]

\newtheorem{proposition}{Proposition}[section]
\newtheorem{remark}{Remark}[section]

\numberwithin{equation}{section}
\usepackage[top=2cm,bottom=2cm,left=1cm,right=1cm]{geometry}

\title{A Generalized Euler probability distribution}
\begin{document}

\maketitle
\begin{center}
\author{Zouha\"ir Mouayn$^{\ast}$ \and Othmane El moize $^{\flat }$}
\end{center}

\begin{abstract}
\scriptsize{From a new class of $q$-deformed coherent states we introduce a generalization of the Euler probability distribution for which the main statistical parameters are obtained explicitly. As application, we discuss the corresponding photon counting statistics with respect to  all parameters labeling the coherent states under consideration.}
\end{abstract}

\begin{center}
\textit{{\footnotesize ${}^{\ast }$ Department of Mathematics, Faculty of Sciences
and Technics (M'Ghila),\\ P.O. Box. 523, B\'{e}ni Mellal, Morocco\vspace*{0.2mm}\\[3pt]
${}^{\flat }$  Department of Mathematics, Faculty of Sciences,\vspace*{-0.1em}\\Ibn Tofa\"{i}l University, P.O. Box. 133, K\'enitra, Morocco}
}
\end{center}



\section{Introduction}
The \textit{canonical} coherent states (CS), denoted $|z\rangle$ and labeled by points $z\in \mathbb{C}$, which go back to the early years of quantum mechanics \cite{schro} may be defined in four ways : (i) as eigenstates of the annihilation operator $A$, (ii) by applying the displacement operator $exp(zA^{*}-\bar{z}A)$ on the vacuum state $|0\rangle$ such that $A|0\rangle =0$, where $A^{*}$ is the Hermitian conjugate of $A$, (iii) by finding states that minimize the Heisenberg uncertainty principle and (iv) as a specific superposition of  eigenstates $\phi_j$ of the harmonic oscillator number operator $N=A^{*}A$ as
\begin{equation}\label{CCS}
\Psi_z=\left( e^{z\bar{z}}\right)^{-1/2} \displaystyle\sum_{j=0}^{+\infty} \frac{\bar{z}^j}{\sqrt{j!}}\phi_j.
\end{equation}

Recently, much attention has been paid to a class of states named "\textit{non classical}" which appeared in quantum optics and in other fields ranging from solid states physics to cosmology. These states exhibit some purely quantum-mechanical properties, such as squeezing and antibunching (sub-Poissonian statistics) \cite{dodo02}. Among them, there are the so-called \textit{generalized} coherent states (GCS) arising by extending one of the four aforementioned  ways defining the canonical CS. These extensions may lead to  inequivalent definitions \cite{pere86}. 

The GCS are usually associated with algebras other than the oscillator one \cite{KLSK85}. An important example is provided by the $q$-deformed CS, related to deformations of the canonical commutation relation or, equivalently, to deformed boson operators \cite{Bi89, ACO}, where this type of CS have been constructed in a way that they reduce to their  standard counterparts as $q\to 1$. Among the latter ones, those satisfying the relation 
\begin{equation*}\label{qcommut}
A_qA_q^{*}-q A_q^{*}A_q=1
\end{equation*}    
with $0<q<1$, see \cite{ACO}. The operators $A_q$ are often termed maths-type $q$-bosons \cite{Solo94} because the \textit{'basic'} numbers and special functions attached to them have been extensively studied in the mathematical literature for a along time \cite{GR}.  A $q$-analogue of the number states expansion \eqref{CCS} is 
\begin{equation}\label{qdefoCS1}
\Psi_z^q:=\left(e_q(z\bar{z})\right)^{-\frac{1}{2}%
}\sum_{j=0}^{+\infty}\frac{\bar{z}^{j}}{\sqrt{[j]_q!}}\phi_j^q,
\end{equation}  
$(1-q)z\bar{z}<1$ and $e_q(z\bar{z})$ being a $q$-exponential function (see Eq.\eqref{qexpo1} below). Here, the $q$-number states $\phi_j^q$ are supposed to span a formal quantum Hilbert space $\mathscr{H}_q$.  By averaging the density matrix  $\hat{\rho_j}:=|\phi_j^q\rangle\,\langle \phi_j^q|$ in the system of CS $\Psi_z^q$, we obtain the Husimi $Q$-function \cite{Husimi} $Q_{\hat{\rho_j}}(z):= \mathbb{E}_{\Psi_z^q}\left(\hat{\rho_j}\right)=\langle \hat{\rho_j}\Psi_z^q,\,\Psi_z^q\rangle_{\mathscr{H}_q}.$ Fixing $z$ and varying $j$, the function
\begin{equation}\label{euler def}
j\mapsto Q_{\hat{\rho_j}}(z)=\frac{|z|^{2j}}{[j]_q!} \left(e_q(|z|^2)\right)^{-1},\quad j=0,1,2,\cdots ,
\end{equation}
defines a $q$-deformed Poisson distribution, here denoted by $X\sim\mathscr{P}(\lambda;q)$, with $\lambda =z\bar{z}$ as a parameter. The latter one was first introduced in \cite{Benkha88} under the name \textit{Euler} distribution. It is unimodal, have increasing failure rates, overdispersed and is infinitely divisible  \cite{Kem92}.  Note that it reduces to the standard Poisson distribution $\mathscr{P}(\lambda)$ in the limit $q\to 1$.\medskip\\
In the present paper, we consider the following generalization of \eqref{qdefoCS1}:
\begin{equation*}
\Psi_{z,m}^q:=\left(\mathcal{N}_{q,m}(z\bar{z}) \right) ^{-\frac{1}{2}}%
\displaystyle \sum_{j=0}^{+\infty }\overline{C_{j}^{q,m}(z)}\phi_j^q,
\end{equation*}
where 
\begin{equation*}
C_j^{q,m}(z):=\frac{(-1)^{min(m,j)} (q;q)_{max(m,j)}q^{\binom{min(m,j)}{2}}\sqrt{1-q}^{|m-j|}|z|^{|m-j|}e^{i(m-j)arg(z)}}{(q;q)_{|m-j|}\sqrt{q^{mj}(q;q)_m(q;q)_j}}P_{min(m,j)}\left( (1-q)z\bar{z};q^{|m-j|}|q\right),
\end{equation*}
are given in terms of Wall polynomials $P_n(\cdot,a|q)$ (\cite{KS}, p.260). For fixed $m\in\mathbb{N}$, these coefficients generalize the above ones $\frac{1}{\sqrt{[j]_q!}}z^j$ in the number states expansion \eqref{qdefoCS1}. $\mathcal{N}_{q,m}(z\bar{z})$ is a factor ensuring the normalization condition   $\langle \Psi_{z,m}^q,\Psi_{z,m}^q\rangle_{\mathscr{H}_q}=1$.

Here, also by the same procedure leading to \eqref{euler def}, we fix $z$, and we associate to the generalized CS $\Psi_{z,m}^q$ a discrete probability distribution, denoted by $X\sim \mathscr{P}(\lambda;q,m)$, as  
\begin{equation}
j\mapsto \mathrm{Pr}(X=j)=\langle \hat{\rho_j}\Psi_z^{q,m},\,\Psi_z^{q,m}\rangle_{\mathscr{H}_q},\quad j=0,1,2,\ldots,\quad \lambda=z\bar{z},
\end{equation}
which naturally, generalizes the Euler probability distribution in \eqref{euler def}. Precisely, our goal is to  introduce explicitly the discrete random variable $\mathscr{P}(\lambda;q,m)$. Next, we  write down its generating function (p.g.f) from which we drive the main statistical parameters of the $q$-deformed  number operator $N_q=A^*_qA_q$. For the latter one, we examine  the photon counting statistics  in the state $\Psi_{z,m}^q$ with respect to the location of the labeling point  $z$ in the corresponding phase space. 

The paper is organized as follows. In Section 2, we recall some notations of $q$-calculus. In Section 3, we define a new class of generalized coherent states after a brief review of  the coherent state formalism and its $q$-analog.  In Section 4, we introduce  a generalized Euler probability distribution and we give a formula for its generating function (p.g.f). In Section 5, the main statistical parameters of the $q$-deformed number operator are derived. Section 6 is devoted to discuss  the domain of classicality/non-classicality of the  generalized $q$-deformed CS we have introduced. In section 7, we conclude with some remarks.
\section{Notations}
  This section collects the basic notations of $q$-calculus and definitions used in the rest of the paper.  For more details we refer to \cite{GR, KS, TE}. We assume that $0<q<1$.\smallskip\\
\textbf{1.} For $a\in\mathbb{C}$, the number 
\begin{equation}\label{q-number}
[a]_q=\frac{1-q^a}{1-q}
\end{equation} 
 is called a $q$-number which satisfies $[a]_q\to a$, $q\to 1$. In particular, $[n]_q$ is called a $q$-integer.\smallskip\\
\textbf{2.} For $a\in\mathbb{C}$, the $q$-shifted factorial is defined by 
  \begin{equation}
(a;q)_{0}=1,\quad (a;q)_{n}:=\prod\limits_{k=0}^{n-1}\left( 1-aq^{k}\right),\:\: n\in \mathbb{N}, \quad (a;q)_{\infty}:=\prod\limits_{k=0}^{\infty}\left( 1-aq^{k}\right),
\end{equation}
and, for any  $\alpha\in\mathbb{C}$, we shall also use
\begin{equation}\label{id11} 
(a;q)_{\alpha}=\frac{(a;q)_{\infty}}{(aq^{\alpha};q)_{\infty}}, \; aq^{\alpha}\neq q^{-n},\, n\in\mathbb{N}.
\end{equation}
\smallskip\\
\textbf{3.} For $a\in\mathbb{C}$, the $q$-shifted factorials satisfy the following useful identities
  \begin{equation}\label{id1}
  (aq^k;q)_{n-k}=\frac{(a;q)_n}{(a;q)_k},\; k=0,1,2,...,n,
  \end{equation}
  \begin{equation}\label{id2}
  (q^{-n};q)_k=\frac{(q;q)_n}{(q;q)_{n-k}}(-1)^k
  q^{\binom{k}{2}-nk},\; k=0,1,2,...,n,
  \end{equation}
  \begin{equation} \label{eq13}
  (aq^n;q)_k=\frac{(a;q)_k(aq^k;q)_n}{(a;q)_n},
  \end{equation}
  \begin{equation}\label{id14}
  (a;q)_{n+k}=(a;q)_n(aq^n;q)_k,
  \end{equation}
  \begin{equation}\label{id15}
  (a;q)_n=(a^{-1}q^{1-n};q)_n(-a)^nq^{\binom{n}{2}},
  \end{equation}
  \begin{equation}\label{id16}
  (aq^{-n};q)_k=\frac{(a^{-1}q;q)_n}{(a^{-1}q^{1-k};q)_n}(a;q)_kq^{-nk},\quad a\neq 0.
  \end{equation}
  \textbf{4.} For $a_1,a_2,...,a_l\;\in\mathbb{C}$, the multiple $q$-shifted factorials is defined as follows
\begin{equation}
(a_1,a_2,...,a_l;q)_n=(a_1;q)_n(a_2;q)_n\cdots(a_l;q)_n,\quad, \:l\in\mathbb{N}\: \mathrm{and}\: n\in\mathbb{N}\cup\{\infty\}.
\end{equation}
\smallskip\\
\textbf{5.} The $q$-binomial coefficient is  defined  by%
\begin{equation}\label{binomid}
\begin{bmatrix} n\\ k \end{bmatrix}_q:=\frac{[n]_{q}!}{[n-k]_{q}![k]_{q}!}=\frac{(q;q)_{n}}{%
(q;q)_{n-k}(q;q)_{k}},\text{ } \: k=0,1,\cdots,n,
\end{equation}%
where 
\begin{equation}\label{qfactor}
[n]_q!=\frac{(q;q)_n}{(1-q)^n}
\end{equation}
denotes the $q$-factorial of  $n$. Eq.\eqref{binomid} can be generalized as 
\begin{equation}\label{binomger}
\begin{bmatrix} \alpha\\ k \end{bmatrix}_q:=(-1)^kq^{k\alpha-\binom{k}{2}}\frac{(q^{-{\alpha}};q)_k}{(q;q)_k},\:\alpha\in\mathbb{C}.
\end{equation}
\smallskip\\
\textbf{6.} There are two different natural $q$-extensions for the exponential function :
\begin{equation}\label{qexpo1}
e_q(\xi):= \displaystyle\sum_{n\geq 0}  \frac{\xi^n}{[n]_q!}=\frac{1}{((1-q)\xi;q)_{\infty}},\quad|\xi|<\frac{1}{1-q},
\end{equation}
and 
\begin{equation}\label{qexpo2}
 E_q(\xi):= \displaystyle\sum_{n\geq 0}  \frac{q^{\binom{n}{2}}}{[n]_q!}\xi^n=(-(1-q)\xi;q)_{\infty}
\end{equation}
which  are related by
\begin{equation}\label{relaqexpo}
e_q(\xi)E_q(-\xi)=1.
\end{equation}
\smallskip\\
\textbf{7.} The  hypergeometric series is defined by 
  \begin{equation}\label{hgfclass}
   \setlength\arraycolsep{1pt}
{}_r F_s\left(\begin{matrix}a_1,...,a_r \\ b_1,...,b_s \end{matrix}\left|\xi\right.\right):=\displaystyle\sum_{k=0}^{\infty} \frac{(a_1)_k\cdots (a_r)_k}{(b_1)_k\cdots(b_s)_k}\frac{\xi^k}{k!},
  \end{equation}
  where $(a)_0:=1 \,\mathrm{and}\: (a)_k:=a(a+1)(a+2)\cdots(a+k-1),\, k=1,2,3,\cdots,$ is the shifted factorial. A basic hypergeometric series is defined by 
  \begin{equation}\label{hgdef}
  \setlength\arraycolsep{1pt}
{}_r \phi_s\left(\begin{matrix}a_1,...,a_r \\ b_1,...,b_s \end{matrix}\left|q;\xi\right.\right):=\displaystyle\sum_{k=0}^{\infty} \frac{(a_1,...,a_r;q)_k}{(b_1,...,b_s;q)_k}(-1)^{(1+s-r)k}q^{(1+s-r)\binom{k}{2}}\frac{\xi^k}{(q;q)_k}.
  \end{equation}
  We note that both of the series ${}_r F_s$ and ${}_r \phi_s$ converges absolutely for all $\xi$ if $r\leq s$ and for $|\xi|<1$ if $r=s+1$. In this special case,  Eq.\eqref{hgdef} reduces to 
  \begin{equation}\label{hgdefsimpl}
  \setlength\arraycolsep{1pt}
{}_{s+1} \phi_s\left(\begin{matrix}a_1,...,a_{s+1} \\ b_1,...,b_s \end{matrix}\left|q;\xi\right.\right):=\displaystyle\sum_{k=0}^{\infty} \frac{(a_1,...,a_{s+1};q)_k}{(b_1,...,b_s;q)_k}\frac{\xi^k}{(q;q)_k}.  \smallskip\\ 
  \end{equation}
  \textbf{8.} The  Laguerre polynomial is defined by
  \begin{equation}\label{Laguerredef}
  L_n^{(\alpha)}(x):=\frac{(\alpha+1)_n}{n!}\sum_{j=0}^n\frac{(-n)_j}{(\alpha+1)_j}\frac{x^j}{j!},\:\alpha>-1,
  \end{equation}
  and can be expressed in terms of the confluent hypergeometric function ${}_1 F_1$ (which is defined by $r=s=1$ in \eqref{hgfclass}) as
  \begin{equation}\label{Laguerredef2}
  L_n^{(\alpha)}(x)=\frac{(\alpha+1)_n}{n!}\setlength\arraycolsep{1pt}
{}_1 F_1\left(\begin{matrix}-n \\ \alpha+1 \end{matrix}\left|x\right.\right).
  \end{equation}
 \textbf{9.} The little $q$-Laguerre or Wall polynomial is defined by means of the basic hypergeometric series ${}_2 \phi_1$ or ${}_2 \phi_0$ as 
  \begin{eqnarray}\label{Wall}
   P_n(x;a|q)=\setlength\arraycolsep{1pt}
{}_2 \phi_1\left(\begin{matrix}q^{-n},0 \\ aq \end{matrix}\left|q;qx\right.\right)
= \frac{1}{(a^{-1}q^{-n};q)_n}\setlength\arraycolsep{1pt}
{}_2 \phi_0\left(\begin{matrix}q^{-n},x^{-1} \\ - \end{matrix}\left|q;\frac{x}{a}\right.\right)
  \end{eqnarray}
  and satisfy 
  \begin{equation}
\displaystyle\lim_{q\rightarrow 1} P_n(x(1-q);q^{\alpha}|q)=\frac{n!}{(\alpha+1)_n}L_n^{(\alpha)}(x).
\end{equation}
\section{A class of generalized $q$-deformed coherent states}
The original idea of coherent states was introduced by E. Schr\"{o}dinger in 1926 \cite{schro} in order to obtain quantum states in $L^2(\mathbb{R})$ that follow the classical flow associated to the harmonic oscillator Hamiltonian 
\begin{equation*}
\hat{H}=-\frac{\hslash^2}{2}\frac{d^2}{dx^2}+\frac{1}{2}x^2-\frac{1}{2}.
\end{equation*}
Namely, we have a set $\{\Psi_z\in L^2(\mathbb{R}),\;z\in\mathbb{C}\}$ labeled by elements of $\mathbb{C}\simeq T^*\mathbb{R}$ (the phase space of a particle moving on $\mathbb{R}$) given by 
\begin{equation}\label{QSdef}
\Psi_z(x)=e^{-\frac{1}{2}\hslash z\bar{z}}\frac{1}{(\pi \hslash)^4 }exp\left(-\frac{1}{2\hslash}(\bar{z}^2+x^2-2\sqrt{2}\bar{z}x)\right)
\end{equation}
where $\hslash$ denotes the Planck parameter. If we denote 
$\{\phi_j\}$ an orthonormal  basis of $L^2(\mathbb{R})$ 
consisting of eigenfunctions of $\hat{H}$ , i.e. $\hat{H}\phi_j=j\phi_j$ (called number states) and we set  $\hslash=1$ for the sake of 
simplicity. Then, the function in \eqref{QSdef} also 
admit an expansion over the basis vectors $\{\phi_j\}$ as
\begin{equation}\label{CSdef}
\Psi_z(x)=\left(e^{z\bar{z}}\right)^{-1/2}\sum_{j=0}^{+\infty} \frac{\bar{z}^j}{\sqrt{j!}}\phi_j(x).
\end{equation}

Similarly to the standard harmonic oscillator, the $q$-analog of the number states are supposed to span a formal quantum Hilbert space $\mathscr{H}_q$ and are given by 
\begin{equation}
\phi_j^q=\frac{(A_q^{*})^n}{\sqrt{[n]_q!}}|0\rangle_q,
\end{equation}
$A_q^{*}$ and $A_q$ are  $q$-creation and $q$-annihilation operators satisfying the commutation relation $A_qA_q^{*}-qA_q^{*}A_q=1$. Here, the oscillator-like Hamiltonian may be defined by $\hat{H}_q:=\frac{1}{2}(A_qA_q^{*}+A_q^{*}A_q)$. Therefore, a generalization of  coherent states \eqref{CSdef} is provided by  states with the form 
\begin{equation}\label{qcsdef2}
\Psi_z^q:=\left(e_q(z\bar{z})\right)^{-\frac{1}{2}%
}\sum_{j= 0}^{+\infty}\frac{\bar{z}^{j}}{\sqrt{[j]_q!}}\phi^q _{j}
\end{equation} 
where $z\bar{z}<(1-q)^{-1}$ which determines the domain of convergence of $e_q(z\bar{z})$. We also observe that  coefficients 
\begin{equation}\label{qcoeffi}
C_{j}^q(z):=\frac{z^{j}}{\sqrt{[j]_q!}},\quad j=0,1,2,\ldots , 
\end{equation}    
appearing in \eqref{qcsdef2} form an orthonormal  system  in the Hilbert space $L^2\left(\mathbb{C},\,d\mu_q(z)\right)$ where
\begin{equation}
\label{dmuq}
d\mu_q(z)=\frac{d\theta}{2\pi}\otimes\sum_{l\geq 0}\frac{q^l(q;q)_{\infty}}{(q;q)_l}d\mu_l(z),
\end{equation}
and $d\mu_l(z)$  is the Lebesgue measure on the circle of radius  $r_l=q^{\tfrac{1}{2}l}/\sqrt{1-q}$. Actually, these coefficients $C_{j}^q(z)$ constitute a particular case of a larger class of $2D$ complex $q$-orthogonal polynomials  introduced in \cite{IZ} as
\begin{equation}\label{2Dqopo}
H_{m,j}(z,\zeta|q):=\displaystyle\sum_{k=0}^{m\wedge j} \begin{bmatrix} m\\k \end{bmatrix}_{q} \begin{bmatrix} j\\k \end{bmatrix}_{q} (-1)^k q^{\binom{k}{2}}(q;q)_k z^{m-k}\zeta^{j-k}, \; z,\zeta\in\mathbb{C}, 
\end{equation}
where $j,m\in\mathbb{N}$ and $m\wedge j=\min(m,j).$  Indeed, by taking $\zeta=\bar{z}$ in \eqref{2Dqopo} and  making a slight modification, we will consider, as in \cite{SOZ18}, the functions 
\begin{equation}
\label{Hmj}
C_j^{q,m}(z):=\frac{(-1)^{m\wedge j} (q;q)_{m \vee j}q^{\binom{m\wedge j}{2}}\sqrt{1-q}^{|m-j|}|z|^{|m-j|}e^{i(m-j)arg(z)}}{(q;q)_{|m-j|}\sqrt{q^{mj}(q;q)_m(q;q)_j}}P_{m\wedge j}\left( (1-q)z\bar{z};q^{|m-j|}|q\right),
\end{equation} 
$P_n(\cdot,a|q)$ denote a Wall  polynomial (\cite{KS}, p.109) and $m\vee j=\max(m,j)$. These coefficients generalize the above ones in \eqref{qcoeffi} in the sense that for $m=0$ we have that $C_j^{q,0}(z)=C_j^{q}(z)$. Therefore, by fixing $m\in\mathbb{N}$ in the expression given by \eqref{Hmj},  we can extend the expression \eqref{qcsdef2} of the $q$-deformed CS  by setting  
\begin{equation}
\Psi_{z,m}^q:=\left(\mathcal{N}_{q,m}(z\bar{z}) \right) ^{-\frac{1}{2}}%
\displaystyle \sum_{j=0}^{+\infty }\overline{C_{j}^{q,m}(z)}\phi_j^q, \label{Eq:phizqm}
\end{equation}
where 
\begin{equation}
\mathcal{N}_{q,m}(z\bar{z})=\frac{(q^{1-m}(1-q)z\bar{z};q)_{m}}{ q^m(q^{-m} (1-q)z\bar{z};q)_\infty},\quad z\bar{z}<q^m(1-q)^{-1}, 
\end{equation}
is a factor ensuring the normalization condition $\langle \Psi_{z,m}^q,\Psi_{z,m}^q\rangle_{\mathscr{H}_q}=1$. These $q$-deformed CS will, naturally, allows us to generalize the Euler distribution \cite{Benkha88} with respect to the parameter $m=0,1,2,\cdots$ . 
\begin{remark}
In \cite{IZ}, Ismail and Zhang have also introduced a second class of 2D complex $q$-orthogonal polynomials by setting 
\begin{equation}
h_{m,j}(z,\zeta|q)=q^{mj}z^m\zeta^j\displaystyle\sum_{k=0}^{+\infty} \frac{(q^{-m},q^{-j};q)_k}{(q;q)_k}\left(\frac{-q}{z\zeta}\right)^k
\end{equation}
which are connected to the previous polynomials \eqref{2Dqopo} by
\begin{equation}
h_{m,j}(z,\zeta|q^{-1})=q^{-mj} i^{-m-j} H_{m,j}(i\,z,i\zeta|q).
\end{equation}
This last relation may also suggest the construction of another extension of  $q$-deformed CS via the same procedure. In a such case, this would lead to a new generalization (with respect to $m$) of the Heine probability distribution \cite{Benkha88}. 
\end{remark}
\section{A generalized Euler distribution}
Averaging the density matrix $\hat{\rho}_j=|\phi_j^q\rangle\langle\phi_j^q|$ in the system of generalized CS $\Psi_{z,m}^q$ enables us to define, in the usual way, a  photon counting probability distribution by setting $p_j(\lambda; q,m):=\langle \hat{\rho_j}\Psi_z^{q,m},\,\Psi_z^{q,m}\rangle_{\mathscr{H}_q},\,\lambda=z\bar{z}$. Explicit calculations lead to the following definition.
\begin{definition}
\textit{\label{def3.3} The discrete random
variable  having the probability distribution 
\begin{equation}
p_j(\lambda; q,m):=\frac{q^{2\binom{m\wedge j}{2}}(1-q)^{|m-j|}\lambda^{|m-j|}}{\mathcal{N}_{m,q}(\lambda)q^{mj}(q;q)_j(q;q)_m}\left(\frac{(q;q)_{m\vee j}}{(q;q)_{|m-j|}}P_{m\wedge j}\left((1-q)\lambda;q^{|m-j|}|q\right)\right)^2,\quad j=0,1,\ldots ~,
\label{Eq:pjmq}
\end{equation}
is denoted by $X\sim \mathscr{P}( \lambda;q,m)$ and will be called a generalized Euler probability distribution of index $m$.}
\end{definition}
Note that for $m=0$, Eq.\eqref{Eq:pjmq} reduces to
\begin{equation}
p_{j}\left( \lambda;q ,0\right) =\frac{\lambda ^{j}}{[j]_q!}E_q(-\lambda),\quad
j=0,1,2,\ldots,
\end{equation}
which is the standard Euler distribution $\mathscr{P}(\lambda;q)$ with parameter $\lambda$ (see \cite{Benkha88}). The $q$-exponential function $E_q(\cdot)$ is defined by \eqref{qexpo2}.
\begin{remark}
In (\cite{jing94}, p. 495) the author has merged into one formula :
\begin{equation}\label{eulersicong}
P(j;\alpha,q)=\frac{1}{[j]_q!}\left(\frac{\alpha}{1-q}\right)^j\left(e_q(\frac{\alpha}{1-q})\right)^{-1}
\end{equation}
the probability mass functions (PMF) of two distributions. That is, for $0<q<1$ and $0<\alpha<1$ with $\alpha=(1-q)\lambda$, Eq.\eqref{eulersicong} reproduces  the Euler distribution \eqref{euler def} while for $q>1$ and $\alpha<0$ it defines the Heine distribution \cite{Benkha88}. 
\end{remark}
 For $m\neq 0$,   one can check that, as $q\to1$, the following limit  
\begin{equation}
p_j(\lambda;q,m) \longrightarrow\frac{\lambda ^{|m-j|}e^{-\lambda }}{j!}%
\left( \frac{(m\wedge j)!\,L_{m\wedge j}^{\left( |m-j|\right) }\left(
\lambda \right) }{\sqrt{m!}}\right) ^{2}\qquad j=0,1,2,\cdots,
\label{Eq:pjmb}
\end{equation}
holds true, $L^{(\alpha)}_k(\cdot)$   being the Laguerre polynomial in \eqref{Laguerredef}.  So that we recover the generalized Poisson distribution $\mathscr{P}(\lambda;m)$ having the quantity in the R.H.S of \eqref{Eq:pjmb} as its mass probability function  \cite{Moutouh}.

 A convenient way to summarize all  properties of the distribution $X\sim\mathscr{P}(\lambda;q,m)$, $m\geq 0$, is with the probability generating  function (p.g.f) which is  defined by the expectation  
\begin{equation}
\mathcal{G}_{X}(t):=\mathbb{E}\left( t^{X}\right),   \label{Eq:defcar}
\end{equation}
where $t$ is a real number. We precisely establish the following result (see Appendix A for the proof).\medskip\\
\textbf{Proposition 4.1.}\label{1} \textit{For $\lambda\in]0,q^m(1-q)^{-1}[$, the p.g.f of $X\sim\mathscr{P}(\lambda;q,m)$ is given by }
\begin{equation}
\mathcal{G}_{X}(t)=\frac{t^m(q^{-m}(1-q)\lambda;q)_\infty}{(q^{-m}t\lambda(1-q);q)_\infty}\,{}_3\phi_2\left(\begin{array}{c}q^{-m},q/t,t\\
q^{1-m}(1-q)\lambda,q\end{array}\Big|q; q\lambda(1-q)\right),\:|t|\leq 1,\\  \label{prop1}
\end{equation}
\textit{in terms of the basic terminating hypergeometric  series ${}_3\phi_2$}.
\begin{corollary}
For $m=0,$ the expression of $\mathcal{G}_{X}(t)$ reduces to 
\begin{equation}
\frac{((1-q)z\bar{z};q)_\infty}{((1-q)z\bar{z}t;q)_\infty}, 
\end{equation}
which is the well known p.g.f of the Euler distribution. For $m\neq 0$, we have the following limit
\begin{equation}\label{PgfClass}
\mathcal{G}_{X}(t)\longrightarrow \;t^m \mathrm{exp}\left(\lambda(t-1)\right)L_m^{(0)}\left(\lambda(2-(t+\frac{1}{t}))\right),\qquad |t|\leq 1;
\end{equation}
as $q\to 1$.
\end{corollary}
\hspace*{-1.2em}\textit{Proof .} Assuming that $|t|\leq 1$ and taking into account the relation \eqref{relaqexpo}, we get that 
\begin{equation}
\lim_{q\to1}\frac{(q^{-m}(1-q)\lambda;q)_\infty}{(q^{-m}t\lambda(1-q);q)_\infty}=\lim_{q\to1}e_q(q^{-m}t\lambda)E_q(-q^{-m}\lambda)=\mathrm{exp}\left(\lambda(t-1)\right).
\end{equation}
By another side, using Eq.\eqref{hgdefsimpl} together with the fact that $(q^{-m};q)_k=0, \forall k>m$, the series ${}_3\phi_2$ in  \eqref{prop1} terminates as 
\begin{eqnarray}\label{terseri}
\sum_{k=0}^{m}\frac{(q^{-m},q/t,t;q)_k}{(q^{1-m}(1-q)\lambda,q;q)_k}\,\frac{\left( q^{-1}\lambda(1-q)\right)^k}{(q;q)_k}.
\end{eqnarray}
Thus, from  identities \eqref{binomger}-\eqref{qfactor} we, successively, have 
\begin{eqnarray}\label{limgfunc}
\lim_{q\to 1}\sum_{k=0}^{m}\frac{(q^{-m},q/t,t;q)_k}{(q^{1-m}(1-q)\lambda,q;q)_k}\,\frac{\left( q^{-1}\lambda(1-q)\right)^k}{(q;q)_k}&=&\sum_{k=0}^{m}\lim_{q\to 1}\left(\frac{(q^{-m};q)_k}{(q;q)_k}\frac{(q/t,t;q)_k}{(q^{1-m}(1-q)\lambda;q)_k}\,\frac{\left( (1-q)\right)^k}{(q;q)_k}\,q^{-k}\lambda^k\right)\cr
&=&\sum_{k=0}^{m}\lim_{q\to 1}\left(\begin{bmatrix} m\\ k \end{bmatrix}_q (-1)^kq^{\binom{k}{2}-mk}\frac{(q/t,t;q)_k}{(q^{1-m}(1-q)\lambda;q)_k}\,\frac{q^{-k}\lambda^k}{[k]_q!}\right)\cr
&=&\sum_{k=0}^{m}\begin{pmatrix} m\\ k \end{pmatrix}(-1)^k\frac{\left(\lambda(1-t)(1-t^{-1})\right)^k}{k!}.
\end{eqnarray}
Finally, from  \eqref{Laguerredef} we can see that the last  sum in \eqref{limgfunc} is the evaluation of the Laguerre polynomial $L_m^{(0)}$ at  $\lambda\left(2-(t+t^{-1})\right)$. This ends the proof. $\square$\begin{remark} By setting $t=e^{iu}$ in the R.H.S of \eqref{PgfClass}, we recover the characteristic function $\Phi_X^m(u)$ of the generalized Poisson distribution $\mathscr{P}(\lambda; m)$, which was obtained in (\cite{Moutouh}, Prop 4.1, p.264).
\end{remark}
\section{Expectation and variance  of the  operator $[N]_q$}
Here, we first start by observing that the expectation and the variance of the $q$-deformed number operator $[N]_q$ in the state $\Psi_{z,m}^q$ coincide with those of the $q$-deformed random variable  whose values are the $q$-numbers $[j]_q,\: j=0,1,2,\cdots,$ as defined  by \eqref{q-number}. Namely, we have 
\begin{equation}\label{Eqdef}
\langle [N]_q \rangle=\sum_{j\geq 0}[j]_qp_j(\lambda;q,m).
\end{equation}
\begin{proposition}
The mean value and the mean square deviation of the number operator $[N]_q=A_q^*A_q$ in the state $\Psi_{z,m}^q$
\label{corollary}  are respectively given by 
\begin{eqnarray}
\langle [N]_q \rangle &=& \lambda+[m]_q,  \label{Eq:Egpd} \\
\langle ([N]_q-\langle [N]_q \rangle)^2  \rangle &=&\lambda^2q^m(1+q-2q^{-m})+\lambda q^m(2[m]_q+q^m).
\label{Eq:Vgpd}
\end{eqnarray}
\end{proposition}
\begin{proof}
From \eqref{Eqdef} and  \eqref{q-number}, one can write
\begin{eqnarray}\label{Eqeq2}
\langle [N]_q \rangle =\frac{1}{1-q}\left(1-\mathcal{G}_{X}(q)\right)
\end{eqnarray}
in terms of the p.g.f \eqref{prop1}. In view of Proposition 4.1, we get that 
\begin{equation}\label{Eqeq1}
\langle [N]_q \rangle =\frac{1}{1-q}\left(1-\frac{q^m(q^{-m}\lambda(1-q);q)_{\infty}}{(q^{1-m}\lambda(1-q);q)_{\infty}}\,{}_3\phi_2\left(\begin{array}{c}q^{-m},1,q\\
q^{1-m}(1-q)\lambda,q\end{array}\Big|q; q\lambda(1-q)\right)\right).
\end{equation}
Now, by using \eqref{terseri} for $t=q$, the series ${}_3\phi_2$  in \eqref{Eqeq1} reduces to 1. On the other hand, by applying the identity \eqref{id11}, Eq.\eqref{Eqeq1} reduces to   
\begin{eqnarray}
\langle [N]_q \rangle&=&\frac{1}{1-q}\left(1-q^m(q^{-m}\lambda(1-q);q)_1\right)=[m]_q+\lambda.
\end{eqnarray}
For the mean square deviation, Eq.\eqref{Eqdef} leads to  
\begin{equation}\label{vardef}
\langle ( [N]_q-\langle [N]_q \rangle)^2\rangle =\sum_{j\geq 0}([j]_q-\langle [N]_q \rangle)^2\:p_j(\lambda;q,m)=\sum_{j\geq 0}[j]_q^2p_j(\lambda;q,m)-\langle [N]_q \rangle^2,
\end{equation} 
and by using the identity \eqref{q-number}, we, successively, obtain 
\begin{eqnarray}\label{vareq1}
\sum_{j\geq 0}[j]_q^2p_j(\lambda;q,m)&=&\frac{1}{(1-q)^2}\left(\sum_{j\geq 0}(1-2q^j+q^{2j})p_j(\lambda;q,m)\right)\cr
&=&\frac{1}{(1-q)^2}\left( \sum_{j\geq 0} p_j(\lambda;q,m)-2q \sum_{j\geq 0}q^j p_j(\lambda;q,m)+\sum_{j\geq 0}q^{2j} p_j(\lambda;q,m)\right)\cr
&=&\frac{1}{(1-q)^2}\left(1-2\mathcal{G}_{X}(q)+\mathcal{G}_{X}(q^2)\right).
\end{eqnarray} 
Next, we replace the quantities occurring in the R.H.S of \eqref{vardef} by their respective expressions in Eq.\eqref{vareq1} and  Eq.\eqref{Eqeq2}. This leads, after simplifications, to
\begin{eqnarray}\label{vareq2}
\langle ([N]_q-\langle [N]_q \rangle)^2  \rangle &=&\frac{1}{(1-q)^2}\left(1-2\mathcal{G}_{X}(q)+\mathcal{G}_{X}(q^2)\right)-\frac{1}{(1-q)^2}\left(1-2\mathcal{G}_{X}(q)+\left(\mathcal{G}_{X}(q)\right)^2\right)\cr
&=&\frac{1}{(1-q)^2}\left(\mathcal{G}_{X}(q^2)-\left(\mathcal{G}_{X}(q)\right)^2\right).
\end{eqnarray} 
Finally, a direct evaluation of $\mathcal{G}_{X}$ at $q^2$ using  \eqref{prop1} gives us
\begin{eqnarray}\label{Gq2eq1}
\mathcal{G}_{X}(q^2)=q^{2m}(1-q^{-m}\lambda(1-q))\left(1-q^{1-m}\lambda(1-q)+q^{-m}\lambda[m]_q(1+q)(1-q)^2\right).
\end{eqnarray} 
Summarizing the above calculations, we arrive at the   expression of $\langle ([N]_q-\langle [N]_q \rangle)^2  \rangle$ as announced in \eqref{Eq:Vgpd}. 
\end{proof}
\section{Photon counting statistics for $[N]_q$}
To define a measure of non classicality of a quantum state, one can follow several different approaches. An early attempt was initiated by Mandel  \cite{mand:79a} who investigated radiation fields and introduced the parameter  
\begin{equation}
\mathcal{Q}=\frac{\mathbb{V}ar\left( Y\right) }{\mathbb{E}%
\left( Y\right)}-1,  \label{Eq:mandP}
\end{equation}
to measure the deviation of the photon counting probability distribution $Y$ from the Poisson distribution for which $\mathcal{Q}=0$. If $\mathcal{Q}<0$, then the underlying statistics are said to be sub-Poissonian  and describes
the anti-bunching of the light. Such anti-bunching is an explicit feature of a quantum field and its observation would provide rather direct evidence of existence of \textit{optical} photons.  Super-Poisson statistics
corresponds rather to $\mathcal{Q}>0$ and the bunching phenomenon occurs. It is possible to look on this phenomena as a characteristic quantum feature of \textit{thermal} photons. 

In our context, the Mandel parameter reads 
\begin{equation}
\mathcal{Q}_{m,q}(\lambda):=\frac{\langle ([N]_q-\langle [N]_q \rangle)^2  \rangle-\langle [N]_q \rangle}{\langle [N]_q \rangle}.
\end{equation}
For $m=0$, it reduces to $\mathcal{Q}_{0,q}=(q-1)\lambda\,,\:\lambda=z\bar{z}$ which shows the sub-Poissonian nature of the photon statistics of $[N]_q$ inside the domain $0<z\bar{z}<(1-q)^{-1}$. For $m\neq 0$, we  may use the results of Proposition \,\ref
{corollary} to find out that sign of $\mathcal{Q}_{m,q}$ is the sign of the quantity 
\begin{equation}\label{polylamda}
P(\lambda):=\lambda^2 q^m(1+q-2q^{-m})+\lambda\left(q^m(2[m]_q+q^m)-1\right)-[m]_q
\end{equation}
which is polynomial in the variable $\lambda=z\bar{z}$ and whose discriminant is 
\begin{equation}\label{delta}
\Delta:=q^{2m}(2[m]_q+q^m)^2+4[m]_q(q^{m+1}-2)-2q^{2m}+1=\frac{q^m-1}{(1-q)^2}\delta
\end{equation}
where
\begin{equation}\label{delta2}
\delta:=(1 + q)^2 q^{3m}+(q-3)(q+1)q^{2m}+(q-1)(3q+1)q^m+7-q(6+q).
\end{equation}
The dependence of the sign of $\delta$ on the  values of $m$ is discussed in Appendix \textbf{B}.\medskip\\
\textbf{Lemma 6.1.} \textit{(i) If $0<q\leq q_0,\:q_0=\frac{5\sqrt{5}-2}{11}$ we have that $\Delta<0$ for all $m\neq 0.$ (ii) If $q_0<q<1$ then :\\
 \hspace*{2 em}(a) $\Delta<0$ if $m>m_q:=\left\lfloor \frac{Log\; \zeta_q}{Log\;q} \right\rfloor$,\\
 \hspace*{2 em}(b) $\Delta>0$ if $m\leq m_q$. Here, $\zeta_q$ is the real solution of the equation
\begin{equation}
(1+q)^2x^3+(q-3)(1+q)x^2+(q(3q-2)-1)x+7-q(6+q)=0
\end{equation}
and $\lfloor s \rfloor$ denotes the greatest integer not exceeding $s$. In this last case, we denote the two roots of the polynomial $P(\lambda)$ in \eqref{polylamda} by }
\begin{equation}
\lambda_\pm^{m,q}=\left(\frac{1-q^m(2[m]_q+q^m)\pm\sqrt{\Delta}}{2q^m(1+q-2q^{-m})}\right)^{1/2}.\medskip\\
\end{equation}
We summarize the discussion on the classicality/nonclassicality of $\Psi_{z,m}^q$ with respect to the location of the labeling point $z$ in the complex plane.\medskip\\
\textbf{Proposition 6.1.} \textit{The photon number statistics for $[N]_q$ in the state $\Psi_{z,m}^q$ are}\smallskip\\
\hspace*{2em}\textit{$(i)$ sub-Poissonian in the following cases :\\
\hspace*{4em}(a) $0<q\leq q_0$ and $m\neq 0$ for $z\bar{z}(1-q)<q^m$.\\
\hspace*{4em}(b) $q_0<q<1$ and $m>m_q$ for $z\bar{z}(1-q)<q^m$.\\ 
\hspace*{4em}(c) $q_0<q<1$ and $m\leq m_q$ for $z\bar{z}<\lambda_+^{m,q}$ or $\lambda_-^{m,q}<z\bar{z}<\frac{q^m}{1-q}$.\\
\hspace*{2em}$(ii)$ super-Poissonian if $q_0<q<1$ and $m\leq m_q$ for $\lambda_+^{m,q}<z\bar{z}<\lambda_-^{m,q}$.\\
\hspace*{2em}$(iii)$ Poissonian if $q_0<q<1$ and $m\leq m_q$ for $z\bar{z}=\lambda_+^{m,q}$ or $z\bar{z}=\lambda_-^{m,q}$}.\\

As we have already mentioned, for $m=0$ the photon statistics of $[N]_q$ is sub-Poissonian inside the whole domain $z\bar{z}<(1-q)^{-1}$. This fact may be known in the literature since it is associated with the Euler distribution. However, when $m\neq 0$ we can conclude from Proposition 6.1 that for specific ranges of parameters $q$ and $m$ namely $q\in]q_0,1[$ and $m\in]0,m_q[$, the $m$-deformation of the Euler distribution as defined by \eqref{Eq:pjmq} gives rise, inside the previous domain $z\bar{z}<(1-q)^{-1}$, to two subdomains where the photon counting of $[N]_q$ exhibits  Poissonian (coherent) and super-Poissonian (thermal) statistics. Since $m_q$ is a threshold value depending on the deformation parameter $q$, we describe below the behavior of $m_q$ with respect to $q$.
\begin{figure}[hbtp]
\centering
\includegraphics[scale=0.5]{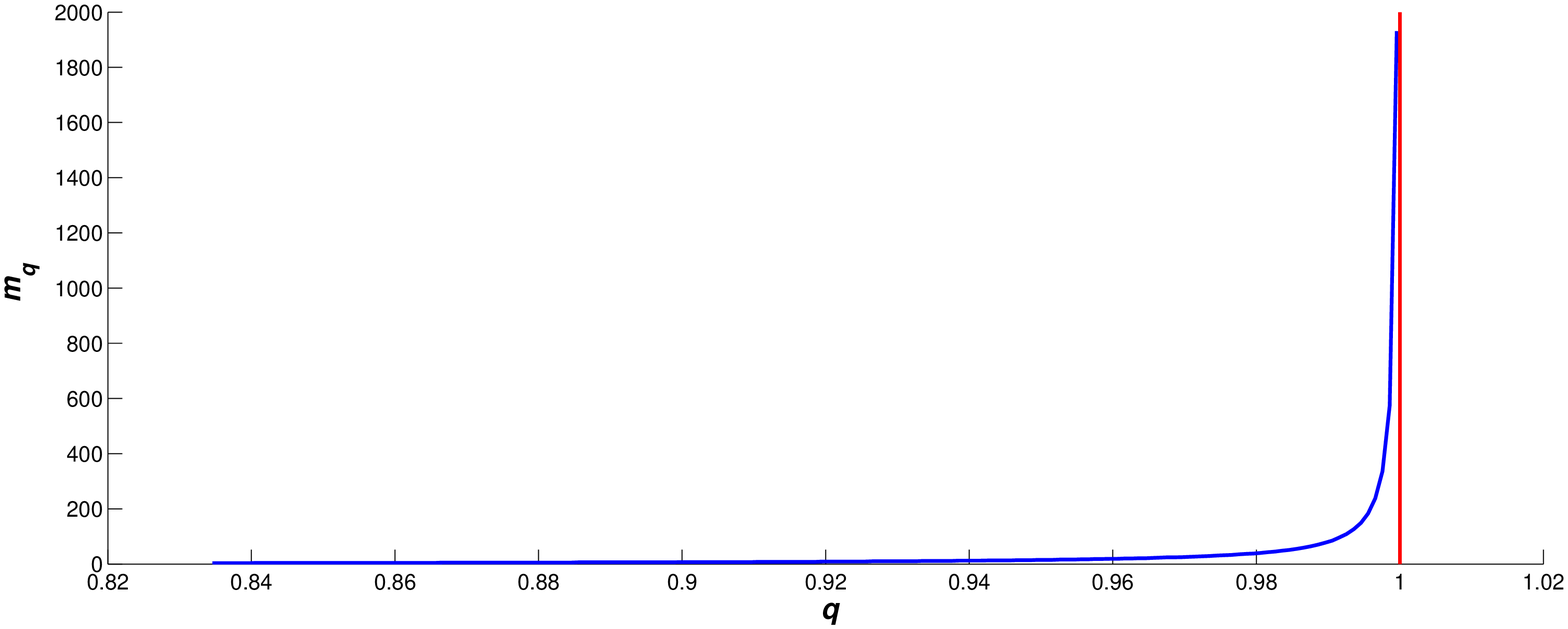}
\end{figure}
\section{Concluding remarks}
While dealing with a new class of $q$-deformed CS, denoted $\Psi_z^{q,m}$, that we have defined $"$\textit{\`a la  Iwata}$"$ by their number states expansion, $0<q<1$ and $z$ being a lebeling complex number, we have introduced a deformation with respect to the parameter $m=0,1,2,\cdots$. For this distribution, we have obtained the generating function from which the main statistical parameters of the $q$-deformed number operator $[N]_q$ have been derived. As application, we have examined the photon counting statistics of  $[N]_q$ in the state $\Psi_z^{q,m}$  with respect to the ranges of parameters $q,z$ and $m$. For $m=0$, theses statistics are sub-Poissonian (antibunching) inside the domain $z\bar{z}<(1-q)^{-1}$ while for $m\neq 0$ specific ranges of $q\in]q_0,1[$ and $m\in]0,m_q[$ reveal that the $m$-deformation of the the Euler distribution gives rise, inside the previous domain of $z$, to the subdomains : circles and annulus where the statistics under consideration are Poissonian (coherent) and super-Poissonian (bunching) respectively. Finally, this analysis may be better understood if we were able to recover our generalized $q$-deformed CS  $\Psi_z^{q,m}$ as a kind of displaced Fock states as $D_q(z)|\phi_m^q\rangle$ where the operator $D_q(z)$ and the ket vector $|\phi_m^q\rangle$ should be fixed. This would require further investigation which will be deferred to a later paper.\medskip\\
\begin{center}
{\large\textbf{Appendix A}}\medskip\\
\end{center}
According to \eqref{Eq:defcar} the p.g.f of $X\sim\mathscr{P}(\lambda;m,q)$ is given by
 \begin{equation}\label{eq01}
 \mathbb{E}\left( t^{X}\right)=\sum_{j=0}^{+\infty}t^jp_j(\lambda,q,m)=\frac{1}{\mathcal{N}_{m,q}(\lambda)(q;q)_m}\mathfrak{S}^{q,m}(\lambda,t)
 \end{equation}
 where 
 \begin{equation}
 \mathfrak{S}^{q,m}(\lambda,t)=\sum_{j=0}^{+\infty}\frac{t^jq^{2\binom{m\wedge j}{2}}(1-q)^{|m-j|}\lambda^{|m-j|}}{q^{mj}(q;q)_j}\left(\frac{(q;q)_{m\vee j}}{(q;q)_{|m-j|}}P_{m\wedge j}\left((1-q)\lambda;q^{|m-j|}|q\right)\right)^2.
 \end{equation}
 We decompose this last sum as 
 \begin{equation}\label{prop1eq01}
 \mathfrak{S}^{q,m}(\lambda,t)=\mathcal{S}_{(<\infty)}^{q,m}(\lambda,t)+\mathcal{S}_{(\infty)}^{q,m}(\lambda,t)
 \end{equation}
 where
 \begin{eqnarray*}
\mathcal{S}_{(<\infty)}^{q,m}(\lambda,t)&=& \sum_{j=0}^{m-1}\frac{t^jq^{2\binom{j}{2}}(1-q)^{m-j}\lambda^{m-j}}{q^{mj}(q;q)_j}\left(\frac{(q;q)_{m}}{(q;q)_{m-j}}P_{j}\left((1-q)\lambda;q^{m-j}|q\right)\right)^2\\
 &-&  \displaystyle\sum_{j=0}^{m-1} \frac{t^jq^{2\binom{m}{2}}(1-q)^{j-m}\lambda^{j-m}}{q^{mj}(q;q)j}\left(\frac{(q;q)_{j}}{(q;q)_{j-m}}P_{m}\left((1-q)\lambda;q^{j-m}|q\right)\right)^2
 \end{eqnarray*}
 and 
 \begin{equation}\label{finisum}
 \mathcal{S}_{(\infty)}^{q,m}(\lambda,t)=\displaystyle\sum_{j=0}^{\infty} \frac{t^jq^{2\binom{m}{2}}(1-q)^{j-m}\lambda^{j-m}}{q^{mj}(q;q)_j}\left(\frac{(q;q)_{j}}{(q;q)_{j-m}}P_{m}\left((1-q)\lambda;q^{j-m}|q\right)\right)^2. 
 \end{equation}
By making use of  the identity  (\cite[p.3]{MoCa}) :
\begin{equation}
 P_n(x;q^{-N}|q)=x^N(-1)^{-N}q^{\frac{N(N+1-2n)}{2}}\frac{(q^{N+1};q)_{n-N}}{(q^{1-N};q)_n} P_{n-N}(x;q^N|q)
 \end{equation}
 for  parameters $N=j-m,n=j$ and $x=(1-q)\lambda$, we obtain that  $\mathcal{S}_{(<\infty)}^{q,m}(\lambda,t)=0.$ For the infinite sum in \eqref{finisum}, let us rewrite  the Wall polynomial  as (\cite{KS}, p.260) :
 \begin{equation}
 P_n(x;a|q)=\frac{(x^{-1};q)_n}{(aq;q)_n} (-x)^n q^{-\binom{n}{2}}{}_2 \phi_1\left(\begin{matrix}q^{-n},0 \\ x q^{1-n} \end{matrix}\left|q;aq^{n+1}\right.\right)
 \end{equation}
where $n=m,\,x=(1-q)\lambda$ and $a=q^{j-m}$. Next, by using  \eqref{id15} and  \eqref{id14} respectively, Eq.\eqref{prop1eq01} becomes
\begin{equation}\label{eq11}
\mathfrak{S}_m^q(t,\lambda)=\frac{q^{2\binom{m}{2}}(q^{1-m}\xi;q)^2_m}{\xi^m} \eta_q^m(t,\lambda)
\end{equation}
where
\begin{equation}\label{eq11.1}
\eta_q^m(t,\lambda)=\sum_{j\geq 0} \frac{Y^j(q;q)_j}{(q;q)_{j-m}^2(q^{j-m+1};q)_m^2}\:\left({}_2\phi_1\left(\begin{array}{c}q^{-m}, 0\\
\xi\, q^{1-m}\end{array}\Big|q; q^{j+1}\right)\right)^2,
\end{equation}
with $\xi=(1-q)\lambda$ and $Y=q^{-m}t\lambda(1-q)$. By identity \eqref{id14}, it follows that $(q;q)_{j-m}^2(q^{j-m+1};q)_m^2=(q;q)_j^2$, and with \eqref{hgdefsimpl}, we get
\begin{eqnarray}\label{eq12}
\eta_q^m(t,\lambda)&=&\sum_{j\geq 0} \frac{Y^j}{(q;q)_j}\sum_{k = 0}^m \frac{(q^{-m};q)_k}{(\xi q^{1-m};q)_k}\,\frac{(q^{j+1})^k}{(q;q)_k} \sum_{l =0}^m \frac{(q^{-m};q)_l}{(\xi q^{1-m};q)_l}\,\frac{(q^{j+1})^l}{(q;q)_l}\cr
&=&\sum_{k,l = 0}^m \frac{(q^{-m};q)_k}{(\xi q^{1-m};q)_k}\,\frac{q^k}{(q;q)_k}  \frac{(q^{-m};q)_l}{(\xi q^{1-m};q)_l}\,\frac{q^l}{(q;q)_l}\,\sum_{j \geq 0} \frac{(q^{k+l}Y)^j}{(q;q)_j}.
\end{eqnarray}
Now, by applying the $q$-binomial theorem (\cite{KS}, p.17):
\begin{equation}
\sum_{n \geq 0} \frac{a^n}{(q;q)_n}=\frac{1}{(a;q)_{\infty}},\quad |a|<1,
\end{equation}
for $a=q^{k+l}\,Y$, which requires the condition $|t|\leq 1$, the R.H.S of \eqref{eq12} takes the form 
\begin{equation}\label{eq13}
\eta_q^m(t,\lambda)= \sum_{k,l = 0}^m \frac{(q^{-m};q)_k}{(\xi q^{1-m};q)_k}\,\frac{q^k}{(q;q)_k}  \frac{(q^{-m};q)_l}{(\xi q^{1-m};q)_l}\,\frac{q^l}{(q;q)_l}\,\frac{1}{(Yq^{l+k};q)_\infty}.
\end{equation}
By applying the identity \eqref{id11} to the factor $\frac{1}{(Yq^{l+k};q)_\infty}$, Eq.\eqref{eq13} can be rewritten as
\begin{eqnarray}\label{dubsum}
\eta_q^m(t,\lambda)&=&\frac{1}{(Y;q)_\infty} \sum_{k,l=0}^m \frac{(q^{-m};q)_k}{(\xi q^{1-m};q)_k}\,\frac{q^k}{(q;q)_k}  \frac{(q^{-m};q)_l}{(\xi q^{1-m};q)_l}\,\frac{q^l}{(q;q)_l}(Y;q)_{l+k}.
\end{eqnarray}
Next, by writing $(Y;q)_{l+k}=(Y;q)_k(q^kY;q)_l$, it follows that
\begin{eqnarray}
\eta_q^m(t,\lambda)&=&\frac{1}{(Y;q)_\infty}  \sum_{k= 0}^m \frac{(q^{-m},Y;q)_k}{(\xi q^{1-m};q)_k}\,\frac{q^k}{(q;q)_k}\,\sum_{l\geq 0} \frac{(q^{-m},q^kY;q)_l}{(\xi q^{1-m};q)_l}\,\frac{q^l}{(q;q)_l}
\end{eqnarray}
which can also be expressed as
\begin{eqnarray}\label{eq132}
\eta_q^m(t,\lambda)&=& \frac{1}{(Y;q)_\infty}  \sum_{k= 0} ^m\frac{(q^{-m},Y;q)_k}{(\xi q^{1-m};q)_k}\,\frac{q^k}{(q;q)_k}\,{}_2\phi_1\left(\begin{array}{c}q^{-m}, q^kY\\
\xi q^{1-m}\end{array}\Big|q; q\right),
\end{eqnarray}
in terms of the series ${}_2\phi_1$. The latter one satisfies the identity (\cite{GR}, p.10) :
\begin{equation}
{}_2\phi_1\left(\begin{array}{c}q^{-n}, b\\
c\end{array}\Big|q; q\right) =\frac{(b^{-1}c;q)_n}{(c;q)_n}b^n,\quad n=0,1,2,...,
\end{equation} 
which, with the parameters $n=m$, $b=q^k\,Y$ and $c=q^{1-m}\xi$, allows us to rewrite \eqref{eq132} as
\begin{eqnarray}\label{eq14}
\eta_q^m(t,\lambda)&=& \frac{1}{(Y;q)_\infty}  \sum_{k= 0}^m \frac{(q^{-m},Y;q)_k}{(\xi q^{1-m};q)_k}\,\frac{q^k}{(q;q)_k}\,\frac{(q^{1-m-k}\xi/Y;q)_m}{(q^{1-m}\xi;q)_m}\,(q^kY)^m\cr
&=&\frac{Y^m}{(Y;q)_\infty(q^{1-m}\xi;q)_m} \sum_{k= 0}^m \frac{(q^{-m},Y;q)_k}{(\xi q^{1-m};q)_k}\,\frac{(q^{1+m})^k}{(q;q)_k} (q^{1-k}/t;q)_m.
\end{eqnarray}
Furthermore, applying the identity \eqref{id16} to $(q^{1-k}/t;q)_m$, gives that 
\begin{eqnarray}\label{eq15}
\eta_q^m(t,\lambda)&=& \frac{Y^m(q/t;q)_m}{(Y;q)_\infty(q^{1-m}\xi;q)_m}{}_3\phi_2\left(\begin{array}{c}q^{-m},Y,t\\
q^{1-m}\xi,q^{-m}t\end{array}\Big|q; q\right).
\end{eqnarray}
 By making appeal to the finite Heine transformation (\cite{GA09}, p.2): 
 \begin{equation}
 {}_3\phi_2\left(\begin{array}{c}q^{-n},\alpha,\beta\\
\gamma,q^{1-n}/\tau\end{array}\Big|q; q\right)=\frac{(\alpha\,\tau;q)_n}{(\tau;q)_n}\;{}_3\phi_2\left(\begin{array}{c}q^{-n},\gamma/\beta,\alpha\\
\gamma,\alpha\,\tau\end{array}\Big|q; \beta\,\tau q^n\right)
 \end{equation}
 for the parameters $\alpha=t,\;\beta=q^{-m}t\xi,\;\gamma=q^{1-m}\xi,\; \tau=q/t$, Eq.\eqref{eq15} reads
\begin{equation}
\eta_q^m(t,\lambda)=\frac{Y^m(q;q)_m}{(Y;q)_\infty(q^{1-m}\xi;q)_m}{}_3\phi_2\left(\begin{array}{c}q^{-m},q/t,t\\
q^{1-m}(1-q)\lambda,q\end{array}\Big|q; q\lambda(1-q)\right).
\end{equation} 
Taking into account the prefactor in \eqref{eq11}, we arrive, after some simplifications, at the expression 
\begin{equation}\label{eq16}
\mathcal{G}_{X}(t)=\frac{t^m(q^{-m}(1-q)\lambda;q)_\infty}{(q^{-m}t\lambda(1-q);q)_\infty}\,{}_3\phi_2\left(\begin{array}{c}q^{-m},q/t,t\\
q^{1-m}(1-q)\lambda,q\end{array}\Big|q; q\lambda(1-q)\right).
\end{equation}
This completes the proof. $\square$\smallskip\\
\begin{center}
{\large\textbf{Appendix B}}\\
\end{center}
In order to determine the sign of $\delta$ in \eqref{delta2} with respect to $m$, we set $\zeta=q^m$ and we look for the solutions  $\zeta$ of the equation
\begin{equation}
(1 + q)^2\zeta^3+(q-3)(q+1)\zeta^2+(q-1)(3q+1)\zeta-q(6+q)+7=0.
\end{equation}
This is equivalent to solve the cubic equation 
\begin{equation}\label{deltamini}
\delta(\zeta)=a\zeta^3+b\zeta^2+c\zeta+d=0
\end{equation}
where  $a=(1+q)^2$, $b=(q-3)(1+q)$, $c=(q-1)(3q+1)$ and $d=7-q(6+q)$. By using the Cardan's method (\cite{Galois}, pp. 4-8), we obtain the following  discriminant
\begin{equation}\label{deltatild}
\tilde{\delta}_q:=\alpha^2+\frac{4\beta^3}{27}
\end{equation}
where $\alpha=\frac{2b^3-9abc+27a^2d}{27a^3}$ and $\beta=\frac{3ac-b^2}{3a^2}$. One can check that $\tilde{\delta}_q\geq 0$ for $0<q<q_0$, $q_0=\frac{5\sqrt{5}-2}{11}$, in this case all real solutions of Eq.\eqref{deltamini} do not belong to the interval $]0,q[$. Therefore, $\delta>0$ for all $ m\neq 0$. From the relation \eqref{delta} and the fact that $\zeta-1<0$, we conclude that $\Delta<0$. On the other hand, for $q_0<q<1$, we have $\tilde{\delta}_q<0$. In this irreducible case, the  roots cannot be extracted directly by Cardan's algebraic formulas. Hence, we arrive at the
so called trigonometrical solution of the cubic equation for the three distinct real roots (\cite{Galois}, pp. 18-19). Here,  only one of them  belonging to the interval $]0,q[$ and it is given by
\begin{equation}
\zeta=\frac{-b}{3a}+2\sqrt{\frac{-\beta}{3}}\,\mathrm{cos}\left(\theta+\frac{4\pi}{3}\right),\qquad \theta=\mathrm{arcos}\left(\frac{3\sqrt{3}\alpha}{2\beta\sqrt{-\beta}}\right).
\end{equation}
Since $\delta<0$ for $\zeta$ such  that  $\zeta_q<\zeta<q$, or equivalently, $m\leq m_q:=\lfloor\frac{Log\;\zeta_{q}}{Log\; q}\rfloor$,  we deduce that $\Delta>0$. Here, $\lfloor s \rfloor$ denotes the greatest integer not exceeding $s$. $\square$\medskip\\
\textbf{\large Acknowledgments.} The authors would like to thank the\textit{ Moroccan Association of Harmonic Analysis $\&$ Spectral Geometry}.



\begin{thebibliography}{99}
\bibitem{schro} E. Schr\"{o}dinger, Die Naturwissenschaften \textbf{14} (1926), 664.
\bibitem{dodo02} V. V. Dodonov, Purity-and entropy-bounded uncertainty relations for mixed quantum states. \textit{J. Opt. B Quantum Semiclass. Opt.}\textbf{ 4} (2002), 98-108.
\bibitem{pere86} A. Perelomov, Generalized coherent states and their applications. Springer-Verlag, Berlin. 1986.
\bibitem{KLSK85} J. R. Klauder, B. S. Skagerstam, Coherent States—Applications in Physics and Mathematics. \textit{Singapore:
World Scientific}. 1985.
\bibitem{Bi89} L. C. Biedenharn, The quantum group {${\rm SU}_q(2)$} and a {$q$}-analogue of the boson operators. \textit{J. Phys. A}. \textbf{22}(18)(1989), 873-878.
\bibitem{ACO} 
  M. Arik, D. D. Coon, Hilbert space of analytic function and generalized coherent states, \textit{J. Math. Phys}. \textbf{17}(4)    (1976), 524-527. 
  \bibitem{Solo94}  A. I. Solomon, Optimal signal-to-quantum noise ratio for deformed photons. \textit{Phys. Lett. A}. \textbf{188} (1994), 215-217.
  \bibitem{GR} G. Gasper, M. Rahman, Basic hypergeometric series. \textit{Cambridge University Press, Cambridge}. \textbf{96}. (2004).
  \bibitem{Husimi}K. Husimi, Quantization of dissipative systems, \textit{Proc. Phys. Soc. Japan}. \textbf{22} (1940), 264.
  \bibitem{Benkha88} L. Benkherouf, J. A. Bather, Oil exploration: sequential decisions in the face of uncertainty. \textit{J. Appl. Probab.} 25 (1988), 529-543.
\bibitem{Kem92} A.W. Kemp, Heine-{E}uler extensions of the {P}oisson distribution. \textit{Comm. Statist. Theory Methods}, \textbf{21}(3)(1992), 571-588. 
\bibitem{IZ} M. E. H. Ismail, R. Zhang, On some $2D$ Orthogonal $q$-polynomials, \textit{Trans. Amer.  Math. Soc} \textbf{369} (10)(2017),  6779-6821.
\bibitem{SOZ18} S. Arjika, O. El Moize, Z. Mouayn. Une $q$-d\'eformation de la transformation de Bargmann vraie-polyanalytique. \textit{C. R.Acad.Sci.Paris}, \textbf{356}(2018), 903-910.
\bibitem{KS} R. Koekoek, R. Swarttouw, The Askey-scheme of hypergeometric orthogonal polynomials and its q-analogues, \textit{Reports of the Faculty of Technical Mathematics and Informatics no. 98-17, Delft University of Technology, Delft,} (1998).
 \bibitem{TE} T. Ernst, A comprehensive treatment of $q$-calculus, Birkh\"auser Springer Basel AG, Basel. (2012).	
\bibitem{Moutouh} Z. Mouayn, A. Touhami,  Probability distributions attached to generalized Bargmann-Fock spaces in the complex plane. \textit{Infin. Dimens. Anal. Quantum Probab. Relat. Top.} \textbf{13} (2010), 257-271.
\bibitem{jing94} Jing Sicong, The $q$-deformed binomial distribution and its asymptotic behaviour. \textit{J. Phys. A}. \textbf{27}(1994), 493-499.
\bibitem{MoCa} S. G. Moreno, C. Garc\'\i, M. Esther, Non-standard orthogonality for the little {$q$}-{L}aguerre polynomials, \textit{Appl. Math. Lett.} \textbf{22} 1745-1749. (2009).
  \bibitem{GA09}
     G. E. Andrews, The finite {H}eine transformation, Combinatorial number theory. (2009), 1-6.  
     \bibitem{mand:79a}  L. Mandel, \newblock Sub-poissonian photon statistics in
resonance fluorescence, \newblock \emph{Opt. Lett.}, \textbf{4} (1979), 205-207.
\bibitem{Galois} D. A. Cox, Galois Theory, Pure and Applied Mathematics, \textit{John Wiley $\&$ Sons} (2nd ed.) 2012.
\end{thebibliography}
\end{document}